\documentclass[10pt]{article}
\usepackage{fullpage}
\usepackage{graphicx}
\usepackage{amsmath}
\usepackage{amsfonts}
\usepackage{amssymb}
\usepackage{amsthm}
\usepackage{algpseudocode}
\usepackage{algorithm}

\begin{document}

\title{\Large\bf Refining Approximating Betweenness Centrality Based on Samplings}
\author{Shiyu Ji, Zenghui Yan\\
\{shiyu,zyan\}@cs.ucsb.edu
}
\date{}
\maketitle

\newtheorem{definition}{Definition}
\newtheorem{theorem}{Theorem}
\newtheorem{lemma}{Lemma}
\newtheorem{corollary}{Corollary}

\newcommand{\p}{\mathrm{Pr}}

\begin{abstract}
Betweenness Centrality (BC) is an important measure used widely in complex network analysis, such as social network, web page search, etc. Computing the exact BC values is highly time consuming. Currently the fastest exact BC determining algorithm is given by Brandes, taking $O(nm)$ time for unweighted graphs and $O(nm+n^2\log n)$ time for weighted graphs, where $n$ is the number of vertices and $m$ is the number of edges in the graph. Due to the extreme difficulty of reducing the time complexity of exact BC determining problem, many researchers have considered the possibility of any satisfactory BC approximation algorithms, especially those based on samplings. Bader et al. give the currently best BC approximation algorithm, with a high probability to successfully estimate the BC of one vertex within a factor of $1/\varepsilon$ using $\varepsilon t$ samples, where $t$ is the ratio between $n^2$ and the BC value of the vertex. However, some of the algorithmic parameters in Bader's work are not yet tightly bounded, leaving some space to improve their algorithm. In this project, we revisit Bader's algorithm, give more tight bounds on the estimations, and improve the performance of the approximation algorithm, i.e., fewer samples and lower factors. On the other hand, Riondato et al proposes a \emph{non-adaptive} BC approximation method with samplings on shortest paths. We investigate the possibility to give an \emph{adaptive} algorithm with samplings on pairs of vertices. With rigorous reasoning, we show that this algorithm can also give bounded approximation factors and number of samplings. To evaluate our results, we also conduct extensive experiments using real-world datasets, and find that both our algorithms can achieve acceptable performance. We verify that our work can improve the current BC approximation algorithm with better parameters, i.e., higher successful probabilities, lower factors and fewer samples.
\end{abstract}

\section{Introduction}
In the past a few decades, how to evaluate the importance of each vertex in a network has attracted extensive researchers \cite{barthelemy2004betweenness}. It is still a popular open problem that how to efficiently find the vertex who plays a central role given the network structure. Many measures of centrality have been proposed: Closeness Centrality \cite{sabidussi1966centrality}, Graph Centrality \cite{hage1995eccentricity}, Stress Centrality \cite{shimbel1953structural}, Betweenness Centrality \cite{freeman1977set}. Recently Betweenness Centrality (BC) has been popularly researched \cite{bader2007approximating, brandes2008variants, abboud2015subcubic}, and has many theoretical applications, e.g., fully dynamic graphs \cite{lee2016efficient}, all pairs shortest paths (APSP) and diameter problems \cite{abboud2015subcubic}, clustering \cite{fairbanks2015behavioral}, etc., as well as many real-world applications, e.g., delay-tolerant networks \cite{magaia2015betweenness}, deltaic river network \cite{cui2015assessment}, hot spots on crime data networks \cite{sivaranjani2015mitigating}, etc.

Betweenness is a measure of centrality based on single source shortest paths (SSSP). We first give its definition.
\begin{definition}
\label{def:bc}
{\bf Betweenness Centrality} \cite{freeman1977set, brandes2001faster, bader2007approximating} of a vertex $v$ in the graph $G=(V,E)$ is
\begin{equation}
BC(v) = \sum_{s,t\in V, s\not=t\not=v} \delta_{st}(v),
\end{equation}
where $\delta_{st}(v)$ is called {\bf pair-dependency} of vertex $v$ given a pair $(s,t)$, and is defined as
\begin{equation}
\delta_{st}(v) = \frac{\sigma_{st}(v)}{\sigma_{st}},
\end{equation}
in which $\sigma_{st}(v)$ denotes the number of the shortest paths from $s$ to $t$ that pass through $v$, and $\sigma_{st}$ denotes the total number of the shortest paths from $s$ to $t$ ($v$ does not have to be on the paths).
\end{definition}
By the definition it is clear that the BC value must be between 0 and $n^2$.

We summarize our major contributions in this paper:
\begin{itemize}
\item We revisit Bader's work and give our new sampling algorithm based on a mixed idea between Bader's and Riondato's works. Analysis and evaluation results are given.
\item We find the bounds in Bader's paper \cite{bader2007approximating} are sub-optimal and some lemmas have gaps, e.g., 1) the proof of Lemma 1 is not given, and 2) in the proof of Lemma 3, the following incorrect argument was used:
$$
\Pr[\left(X_1 - \frac{A}{n}\right) + \cdots + \left(X_k - \frac{A}{n}\right) \geq cn-\epsilon n] 
\leq \sum_i \Pr[X_i - \frac{A}{n} \geq cn-\epsilon n].
$$
Clearly the sum of $n$ random variables exceeds a value $S$ does not imply one of these random variable exceeds this value $S$.
\item To repair the theory based on our findings above, we develop a new probability theoretic model that can well explain the results in Bader's paper. Our results not only agrees with the existing research works mostly, but also gives more refined bounds and better performance, e.g., lower factors and higher successful probabilities.
\end{itemize}

We organize this paper as follows. In the related works, we will discuss the most important existing techniques of finding the exact or estimated BC values given a large graph. Then we will revisit Bader's algorithm and give our new algorithm sampling on pairs of vertices. For both algorithms, we will rigorously show the results that can achieve improved bounds on number of samplings and approximation factors. To evaluate the performance, we will discuss our experimental results in Section \ref{sec:exp}. Finally, future works and conclusions will be presented

\section{Related Works}
Brandes \cite{brandes2001faster, brandes2008variants} gives the state-of-the-art fastest algorithm to calculate the exact value of BC for each vertex in the graph. Given a graph $G=(V, E)$, Brandes' algorithm takes at least $O(nm)$ time for unweighted graphs and $O(nm + n^2 \log n)$ time for weighted graphs, where $n$ is the number of vertices $|V|$ and $m$ is the number of edges $|E|$. However, for very large and complex networks, the performance of Brandes' algorithm may be unacceptable \cite{bader2007approximating}. For example, in large social networks, Brandes' algorithm may take considerably large amount of time, but we want to find the hot spots or trending timely \cite{kourtellis2013identifying}. There are some heuristic variants on Brandes' algorithm, but the complexity has not seen significantly improvement \cite{puzis2012heuristics}. 

To address the performance problem, people have turned to paralleled algorithms \cite{madduri2009faster, bader2008graph, riondato2014fast, RiondatoK16}, randomized approximations \cite{bader2007approximating, geisberger2008better, riondato2014fast, RiondatoK16, C14, ASR14, ASR15}, etc. It turns out that sampling plays an important role in these algorithms. There are two categories on the state-of-art sampling approximation algorithms: uniform sampling (e.g., to sample nodes with identical probability)\cite{bader2007approximating, geisberger2008better, RU16} and non-uniform sampling (e.g., each node has a different probability to be sampled) \cite{riondato2014fast, RiondatoK16, C14, ASR14, ASR15}. However, none of the existing methods can show dominating performance over the rest. Hence how to choose the samplings to achieve better performance is still an open problem \cite{bader2007approximating} Moreover, many BC approximations by samplings lack theoretical performance bounds, and thus an improved bound may also speed up the algorithm \cite{geisberger2008better}. 

In this project, we will revisit Bader's algorithm, which takes $\varepsilon t$ samples to estimate the BC values within a factor of $1/\varepsilon$, where $t$ is the ratio between $n^2$ and the BC value of the vertex. We will discuss the performance bound of the algorithm, and refine the sampling parameters to achieve better results, e.g., fewer samples and lower factors. On the other hand, we will also use the idea of Riondato's sampling on shortest paths \cite{riondato2014fast}. We will propose and analyze a new adaptive BC approximation algorithm. We will also use simulations with real-world data to verify our results.

\section{Brandes' Exact Computation of Betweenness Centrality}
In this section we will talk about the main idea of Brandes' algorithm \cite{brandes2001faster}, which gives exact computation on betweenness centrality (BC) value of each vertex in the network. This is the fastest algorithm on exact BC computation currently.

Dependency of a source vertex on a given node is frequently used in our algorithms and discussions. Here we give the definition: \cite{brandes2001faster}
\begin{definition}
Given a graph $G=(V,E)$, the {\bf dependency} of a source vertex $s\in V$ on a vertex $v\in V$ is
\begin{equation}
\delta_{s\cdot}(v) = \sum_{t\in V, s\not= t\not= v} \delta_{st}(v).
\end{equation}
\end{definition}
Based on Definition \ref{def:bc}, the BC value of vertex $v$ is
\begin{equation}
BC(v) = \sum_{s\not= v} \delta_{s\cdot}.
\end{equation}
Brandes finds a recursive way to calculate the BC value of $v$, by introducing {\it predecessors} of $v$.
\begin{definition}
Given a graph $(V,E)$, the {\bf predecessors} of a vertex $v\in V$ on a shortest path from $s$ to $v$ is a subset $P_s(v)\subseteq V$ s.t.
$$t\in P_s(v) \Rightarrow \left( d(s,v) = d(s,t) + d(t,v) \wedge (t,v)\in E \right),$$
where $d(s,t)$ denotes the length of a shortest path from $s$ to $t$.
\end{definition}
The following theorem motivates Brandes' algorithm. Its proof is in \cite{brandes2001faster}.
\begin{theorem}
\label{thm:dp}
Given a graph $(V,E)$, for any $s,v\in V$, we have
\begin{equation}
\delta_{s\cdot}(v) = \sum_{t\in V s.t. v\in P_s(t)} \frac{\sigma_{sv}}{\sigma_{st}}(1+\delta_{s\cdot}(t)).
\end{equation}
\end{theorem}
The basic idea of Brandes' algorithm can be summerized as follows:
\begin{enumerate}
\item For each vertex $s\in V$, we calculate the shortest paths from $s$ to other vertices, i.e., using Dijkstra's algorithm, which takes at least $O(|E|+|V|\log|V|)$ time.
\item For each vertex $s\in V$, traverse the vertices in descending order of their distances from $s$, and accumulate the dependencies by Theorem \ref{thm:dp}. Each traversal takes $O(|E|)$ time.
\end{enumerate}
The total complexity $T(n,m)$ of Brandes' algorithm, where $n=|V|$ and $m=|E|$, is therefore 
$$T(n) = O(n(m+n\log n)+nm) = O(nm+n^2\log n).$$
Even though Brandes' is the fastest algorithm on finding exact BCs, its time complexity can be unacceptable in many practical scenarios, e.g., very big graphs or large scaled datasets. This motivates the research in BC approximation based on samplings.

\section{Approximation on Betweenness Centrality Based on Adaptive Vertex Samplings}
In this section, we give Bader's BC approximation by adaptive-sampling \cite{bader2007approximating}. The basic idea of Bader's algorithm is that if we are only interested in a loose estimation of BC value, then we do not have to compute the shortest paths from each vertex $s\in V$. Instead we can take samples $S\subseteq V$ ($|S|<<|V|$) and compute the shortest paths on each sample $s\in S$. Then we can accumulate the dependencies like how Brandes' does, and it is guaranteed that the final sum of dependencies is a good estimation of the BC value, i.e., the approximating factor can be provably bounded.

Consider we want to estimate the BC value of a vertex $v$ given the graph $G=(V,E)$. We summerize Bader's algorithm as follows:
\begin{enumerate}
\item Let $S=0$.
\item Sample a vertex $v_i\in V$. Compute the shortest paths from $v_i$. Accumulate the dependencies $S \gets S+ \delta_{v_i\cdot}(v)$.
\item If $S \leq cn$ for some $c\geq 1$\footnote{Bader's paper \cite{bader2007approximating} requires $c\geq 2$. We will see the lower bound of $c$ can be relaxed to 1.}, repeat Step 2. Otherwise, output $nS/k$ as the estimation of $BC(v)$, where $k$ is the number of samples.
\end{enumerate}
Bader shows this algorithm can successfully estimate the BC value with a high probability and its factor has a upper bound. The following theorem is the main result.
\begin{theorem}
\label{thm:bader}
{\bf (Bader's)} For any $\varepsilon \in (0,1/2)$, Bader's algorithm can estimate $BC(v)$ within a factor of $1/\varepsilon$ with a successful probability no less than $1-2\varepsilon$ and $\varepsilon t$ samples, where $t = n^2/BC(v)$.
\end{theorem}
Note that the number of samples in Bader's method is proportional to $t= n^2/BC(v)$, which implies that if the target vertex has extremely low BC, the sampling is likely to run forever.
Thus Bader's method only fits for the vertices with high BCs. There is some recent advance on approximating the vertices with top BCs, e.g., \cite{RU16}.
Hence we assume the given target vertex is unlikely to cause infinite sampling.

Bader's proof uses some loosely bounded arguments, e.g., the upper bound of the variance of the dependency $\delta_{v_i\cdot}(v)$ for each sample $v_i$. In this paper we will discuss the results when the arguments are more tightly bounded. In fact, we will show that the results in Theorem \ref{thm:bader} can be improved.

\section{Refining Results of BC Approximation Samplings}
In this section we will discuss the probabilistic sampling model, and show an improved theorem which achieves better parameters compared to Theorem \ref{thm:bader}.

We first give the model assumptions on the samplings. Let $X_i$ be the random variable of dependency $\delta_{v_i\cdot}(v)$, i.e., the sampled source vertex is $v_i$, and the target vertex is $v$. Let $A=BC(v)$. The following lemma gives the probability distribution of $X_i$.
\begin{lemma}
\label{lem:df}
The cdf of $X_i$ is
$$\p [X_i\leq x] = 
\begin{cases}
0 \quad x\leq 0, \\
1-\left( 1-\frac{x}{A} \right)^{n-1}\quad 0<x<A,\\
1 \quad x\geq A.
\end{cases}$$
Therefore, the pdf of $X_i$ is
$$p(x) =
\begin{cases}
\frac{n-1}{A} \left( 1-\frac{x}{A} \right)^{n-2} \quad 0<x<A,\\
0 \quad \mathrm{otherwise}.
\end{cases}$$
\end{lemma}
\begin{proof}
See the Appendix.
\end{proof}
In the proof of Lemma \ref{lem:df}, we assume the cutting points of the sub-intervals representing each $X_i$ are uniformly distributed. This is the major assumption of our model. We will show that our model is highly compatible to the existing results, e.g., \cite{bader2007approximating}

Now we can compute the mean $E[X_i]$ and the variance $Var[X_i]$. Note a similar lemma has been proposed by Bader's paper (Lemma 2 in \cite{bader2007approximating}).
\begin{corollary}
\label{cor:ev}
$$E[X_i] = \frac{A}{n}, \quad Var[X_i] = \frac{n-1}{n^2(n+1)}A^2.$$
\end{corollary}
\begin{proof}
See the Appendix.
\end{proof}
Note that in \cite{bader2007approximating}, Lemma 2 states
$$E[X_i] = A/n, E[X_i^2] \leq A, Var[X_i] \leq A.$$
Since $A \leq n^2$, it is clear our bounds satisfy the cases of $E[X_i]$ and $Var[X_i]$. We leave the discussion on the case of $E[X_i^2]$ to the Appendix.

Like Bader's paper \cite{bader2007approximating}, we need two lemmas to show our main result.
\begin{lemma}
\label{lem:1}
Let $k=\varepsilon (\frac{n^2}{A})^{\frac{2}{3}}$. Then
$$\p[\sum_{i=1}^k X_i \geq cn] \leq \frac{\varepsilon^3}{(c-\varepsilon)^2}.$$
\end{lemma}
\begin{proof}
See the Appendix.
\end{proof}

\begin{lemma}
\label{lem:2}
Let $k\geq\varepsilon (\frac{n^2}{A})^{\frac{2}{3}}$ and $d>0$. Then
$$\p[\bigg|\frac{n}{k}\sum_{i=1}^k X_i - A\bigg|\geq dA] \leq \frac{1}{\varepsilon d^2} \left(\frac{A}{n^2}\right)^{\frac{2}{3}}.$$
\end{lemma}
\begin{proof}
See the Appendix.
\end{proof}

Now we have our main theorem.
\begin{theorem}
\label{thm:main}
For $0<\varepsilon < 1/2$, if $BC(v) = n^2/t$, $t\geq 1$, then with probability more than $1-\left(1+\frac{1}{(2c-1)^2}\right)\varepsilon$, BC can be estimated within a factor of $1/(\varepsilon t^{1/3})$ with $\varepsilon t^{2/3}$ samples.
\end{theorem}
\begin{proof}
We first estimate the probability that our algorithm terminates with $k=\varepsilon (\frac{n^2}{A})^{2/3}$ samples. Lemma \ref{lem:1} guarantees that this probability is less than $\varepsilon^3/(c-\varepsilon)^2$. Since $\varepsilon < 1/2$, we have
$$\frac{\varepsilon^3}{(c-\varepsilon)^2} \leq \varepsilon \frac{1/4}{(c-1/2)^2}=\frac{\varepsilon}{(2c-1)^2}.$$
We next estimate the probability that our algorithm fails to estimate within the factor. By setting $d=1/(\varepsilon t^{1/3})$, Lemma \ref{lem:2} guarantees that the probability of such failure is at most $\varepsilon$, if the number of samples $k\geq \varepsilon (\frac{n^2}{A})^{2/3}=\varepsilon t^{2/3}$.

Note that except the above two cases, there is no other case when our algorithm will fail to estimate. By subtracting the probabilities, our algorithm will succeed with probability more than $1-\varepsilon-\frac{1}{(2c-1)^2}\varepsilon$.
\end{proof}

By setting $c=1$ in Theorem \ref{thm:main}, we have the following corollary, which has better parameters compared to Theorem 3 in \cite{bader2007approximating}.
\begin{corollary}
For $0<\varepsilon < 1/2$, if $BC(v) =  n^2/t$, $t\geq 1$, then with probability more than $1-2\varepsilon$, BC can be estimated within a factor of $1/(\varepsilon t^{1/3})$ with $\varepsilon t^{2/3}$ samples.
\end{corollary}

\section{Adaptive Samplings on Pairs of Vertices}
In this section we will propose our BC approximation by adaptive samplings on pairs of vertices, and analyze its performance. 

\subsection{The Approximation Algorithm}
Riondato et al. \cite{riondato2014fast, RiondatoK16} give the idea to approximate BC by sampling on shortest paths between a pair of vertices. The number of samplings are calculated beforehand by using Vapnik-Chervonenkis theory. Thus it is {\it not} adaptive. Riondato's approximation can achieve very small error $\varepsilon$ with very high successful probability, i.e., $1-\delta$. However, the number of samplings can be very high, i.e., $\propto (\ln 1/\delta)/\varepsilon^2$. This is because Riondato's algorithm can give simultaneous probabilistic guarantees on the approximation for each vertex, not only one target vertex. Thus it is natural to need more samplings. But it is still interesting to investigate the possibility of adaptive samplings not only on single vertices. Riondato's \textsf{ABRA} algorithm \cite{RU16} gives the idea to sample on pairs of vertices. Our idea is to merge Bader's \cite{bader2007approximating} and Riondato's \cite{riondato2014fast, RiondatoK16, RU16} methods. 
We propose Algorithm \ref{alg} to estimate the BC value of the vertex $t$.

\begin{algorithm}
\centering
\begin{minipage}{0.95\textwidth}
\begin{algorithmic}
\Require Graph $G=(V, E)$. Target vertex $t\in V$.
\Ensure Approximated $BC(t)$.
\State $S \gets 0$, $k \gets 0$.
\While {$S \leq cn$ for some constant $c\geq 1$}
	\State Choose a pair of different vertices $(u, v)$ in the graph.
	\State Find all the shortest paths from $u$ to $v$. Denote by $\sigma_{uv}$ the number of the paths.
	\State Denote by $\sigma_{uv}(t)$ the number of the above paths that pass through $t$.
	\State $S \gets S + \sigma_{uv}(t)/\sigma_{uv}$.
	\State $k \gets k+1$.
\EndWhile
\State Return $n(n-1)S/k$ as the approximated BC of $t$.
\end{algorithmic}
\end{minipage}
\caption{Betweenness Centrality Approximation Algorithm Based on Adaptive Sampling on Pairs of Vertices}
\label{alg}
\end{algorithm}

\subsection{Analysis}
Let random variable $Y_i$ be the value $\sigma_{uv}(t)/\sigma_{uv}$ of one sample $(u, v)$. By a similar manner as $X_i$ in sampling on vertices, we may find the cdf and pdf of $Y_i$.

\begin{lemma}
\label{lem:df2}
The cdf of $Y_i$ is
$$\p [Y_i\leq x] = 
\begin{cases}
0 \quad x\leq 0, \\
1-\left( 1-\frac{x}{A} \right)^{n(n-1)-1}\quad 0<x<A,\\
1 \quad x\geq A.
\end{cases}$$
Therefore, the pdf of $Y_i$ is
$$p(x) =
\begin{cases}
\frac{n(n-1)-1}{A} \left( 1-\frac{x}{A} \right)^{n(n-1)-2} \quad 0<x<A,\\
0 \quad \mathrm{otherwise}.
\end{cases}$$
\end{lemma}

The mean and variance of $Y_i$ can also be found analogously.
\begin{corollary}
\label{cor:ev2}
$$E[Y_i] = \frac{A}{n(n-1)}, \quad Var[Y_i] = \frac{n(n-1)-1}{n^2(n-1)^2(n(n-1)+1)}A^2.$$
\end{corollary}

Like sampling on vertices, we need two lemmas to help prepare the proofs.

\begin{lemma}
\label{lem:3}
Let $k=\varepsilon (\frac{n^2(n-1)}{A})^{\frac{2}{3}}$. Then
$$\p[\sum_{i=1}^k Y_i \geq cn] \leq \frac{\varepsilon^3}{(c-\varepsilon)^2}.$$
\end{lemma}
\begin{proof}
See the Appendix.
\end{proof}

\begin{lemma}
\label{lem:4}
Let $k\geq\varepsilon (\frac{n^2(n-1)}{A})^{\frac{2}{3}}$ and $d>0$. Then
$$\p[\bigg|\frac{n(n-1)}{k}\sum_{i=1}^k Y_i - A\bigg|\geq dA] \leq \frac{1}{\varepsilon d^2} \left(\frac{A}{n^2(n-1)}\right)^{\frac{2}{3}}.$$
\end{lemma}
\begin{proof}
See the Appendix.
\end{proof}

Now we can summarize the performance result of sampling on pairs of vertices.

\begin{theorem}
\label{thm:main2}
For $0<\varepsilon < 1/2$, if $BC(v) = n^2/t$, $t\geq 1$, then with probability more than $1-\left(1+\frac{1}{(2c-1)^2}\right)\varepsilon$, BC can be estimated within a factor of $\frac{1}{\varepsilon}\left(\frac{1}{t(n-1)}\right)^{\frac{1}{3}}$ with $\varepsilon t^{\frac{2}{3}}(n-1)^{\frac{1}{3}}$ samples.
\end{theorem}
\begin{proof}
We first estimate the probability that our algorithm terminates with $k=\varepsilon (\frac{n^2(n-1)}{A})^{2/3}$ samples. Lemma \ref{lem:1} guarantees that this probability is less than $\varepsilon^3/(c-\varepsilon)^2$. Since $\varepsilon < 1/2$, we have
$$\frac{\varepsilon^3}{(c-\varepsilon)^2} \leq \varepsilon \frac{1/4}{(c-1/2)^2}=\frac{\varepsilon}{(2c-1)^2}.$$
We next estimate the probability that our algorithm fails to estimate within the factor. By setting $d=\frac{1}{\varepsilon}\left(\frac{1}{t(n-1)}\right)^{\frac{1}{3}}$, Lemma \ref{lem:2} guarantees that the probability of such failure is at most $\varepsilon$, if the number of samples $k\geq \varepsilon (\frac{n^2(n-1)}{A})^{2/3}=\varepsilon t^{2/3}(n-1)^{\frac{1}{3}}$.
\end{proof}

Similarly, we have the following corollary.
\begin{corollary}
For $0<\varepsilon < 1/2$, if $BC(v) = n^2/t$, $t\geq 1$, then with probability more than $1-2\varepsilon$, BC can be estimated within a factor of $\frac{1}{\varepsilon}\left(\frac{1}{t(n-1)}\right)^{\frac{1}{3}}$ with $\varepsilon t^{\frac{2}{3}}(n-1)^{\frac{1}{3}}$ samples.
\end{corollary}
Compared with sampling on vertices, the factor is much less $\frac{1}{\varepsilon}\left(\frac{1}{t(n-1)}\right)^{\frac{1}{3}} < 1/(\varepsilon t^{1/3})$, but we use more samples. Fortunately, sampling on pairs has less time cost than sampling on vertices, i.e., it is not necessary to calculate the single source shortest paths (SSSP) to all the vertices in the graph. When the destined vertex is reached, the SSSP calculation can stop earlier without exhausting all the vertices in the graph. However, sampling on vertices requires SSSP to all the vertices for each sample. In this light, we believe sampling on pairs of vertices has better performance than sampling on vertices. We will use simulations on real-world data to verify this fact.

\section{Evaluations}
\label{sec:exp}
In this section, we will present our experimental results of revised Bader's method and our proposed algorithms. We will consider the accuracy of the estimation and running time, i.e., the approximation factors and the number of samplings.

{\it Experiment setup}: We implemented our algorithm, and used the implementation of Brandes' \cite{brandes2001faster} from igraph \cite{csardi2006igraph}. All the implementations were written in C, and were tested on uniform system: 4-core Intel Xeon Processors with 16GB RAM, and SMP Linux with kernel version 3.19.8. Each measurement was averaged over 10 instances.

{\it Datasets}: We use Stanford Large Network Dataset Collection (SLNDC) \cite{leskovec2014snap}. We focus one two real-world graphs, \texttt{oregon1-010331} and \texttt{ego-Facebook}. Note that \texttt{oregon1-010331} has been used by Riondato \cite{riondato2014fast}, and \texttt{ego-Facebook} has fewer vertices but denser edges. Table \ref{tab:graphs} gives details about the two graphs.
\begin{table}[!t]
\centering
\begin{tabular}{c || c | c | c}
 & $|V|$ & $|E|$ & avg. clustering coef.\\
\hline
\texttt{oregon1-010331} & 10670 & 22002 & 0.2970 \\
\texttt{ego-Facebook} & 4039 & 88234 & 0.6055 \\
\end{tabular}
\caption{Details from SLNDC.}
\label{tab:graphs}
\end{table}

\subsection{Revised Bader's Algorithm}

\begin{figure*}[!t]
\centering
\begin{minipage}[t]{0.3\textwidth}
\centering
\includegraphics[width=\textwidth]{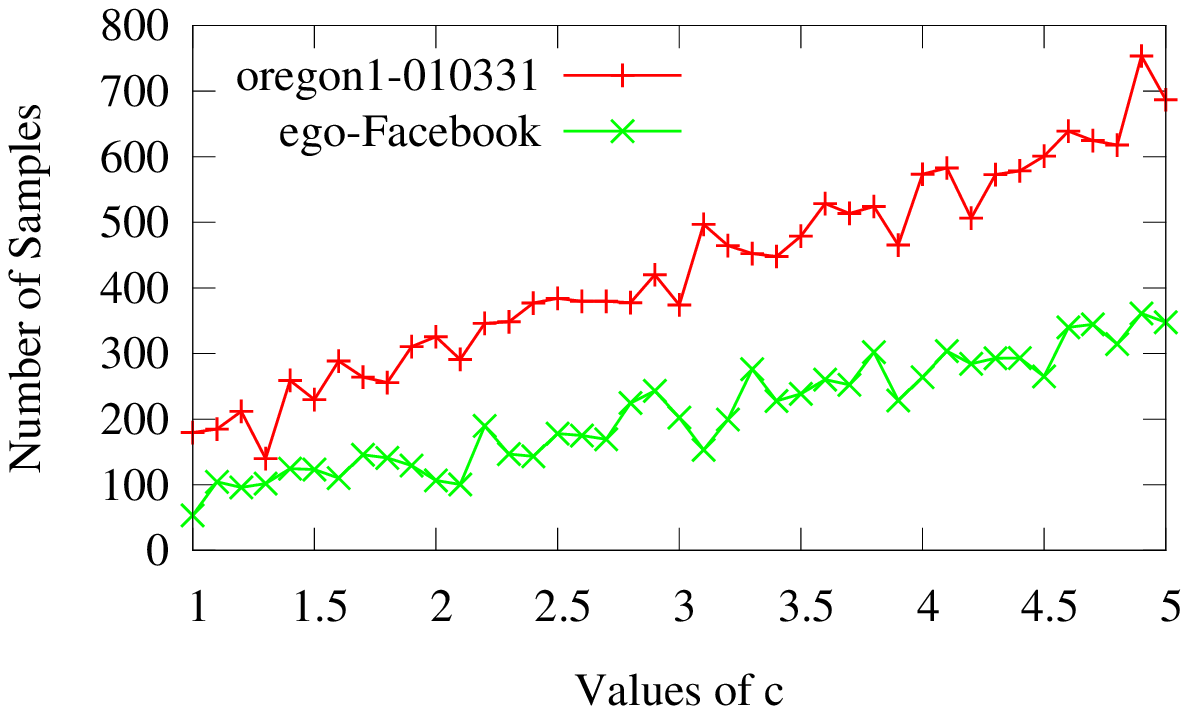}
\caption{$c$ $\sim$ \# samples.}
\label{fig:bader_c_sample}
\end{minipage}
\hspace{10pt}
\begin{minipage}[t]{0.3\textwidth}
\centering
\includegraphics[width=\textwidth]{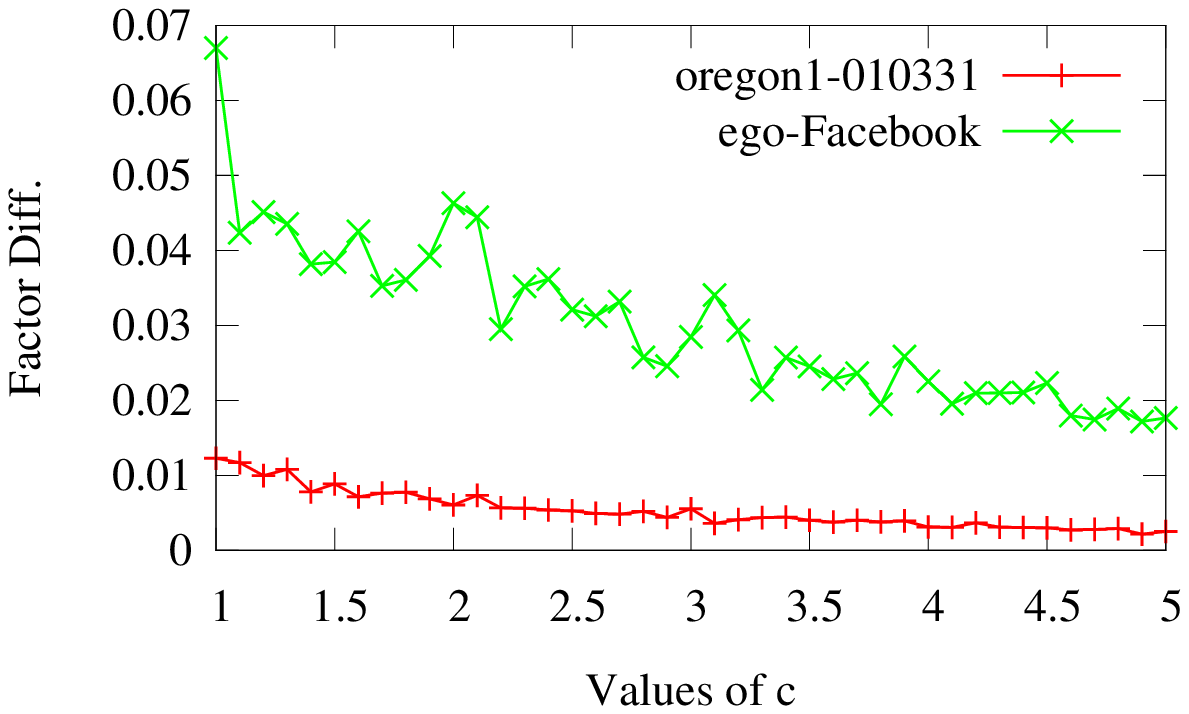}
\caption{$c$ $\sim$ factor diff.}
\label{fig:bader_c_factor}
\end{minipage}
\hspace{10pt}
\begin{minipage}[t]{0.3\textwidth}
\centering
\includegraphics[width=\textwidth]{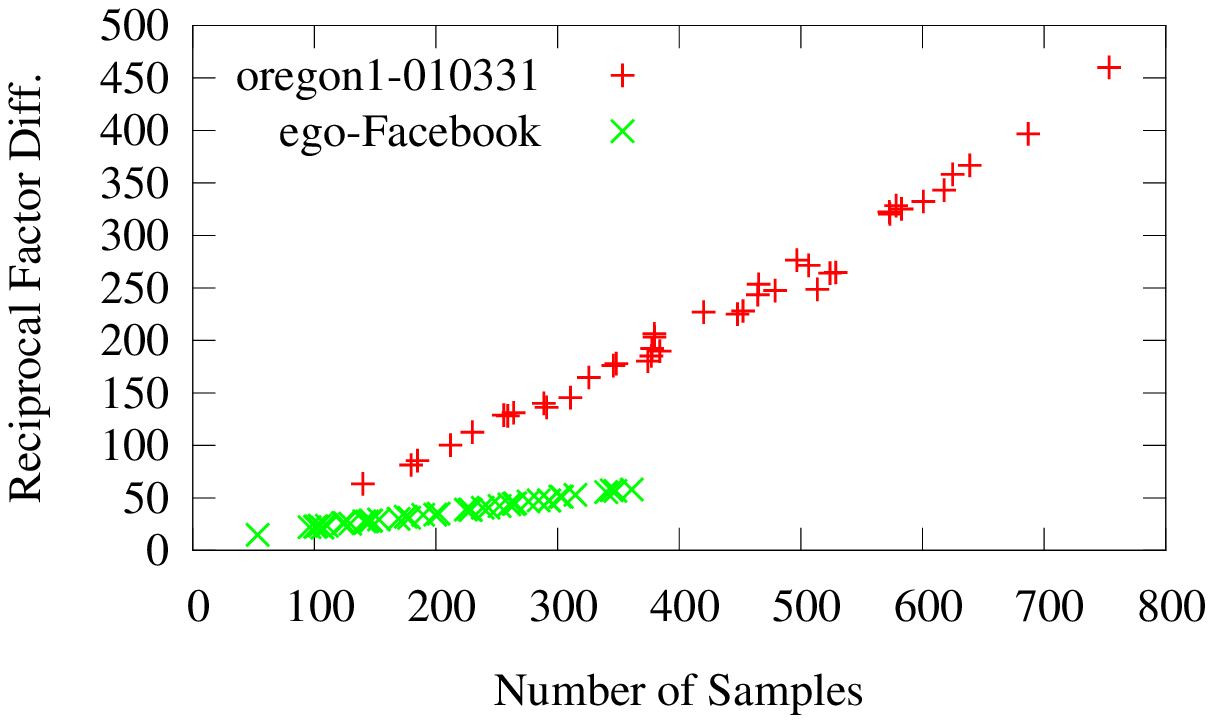}
\caption{\# samples $\sim$ 1/factor diff.}
\label{fig:bader_sample_factor}
\end{minipage}
\end{figure*}

We present the results for our revised Bader's algorithm. 
Figure \ref{fig:bader_c_sample} gives that the number of samplings increases if the parameter $c$ in the algorithm increases. The growth is almost linear, since the curves present steady slopes. This is because in the algorithm, higher $c$ gives larger threshold when the loop terminates, and hence admits more samplings. Note that the range of $c$ starts from 1.0. That is different from \cite{bader2007approximating}, which required $c>2$. The total number of samplings is no more than 800 for both graphs, which is much less than the number of vertices.

Figure \ref{fig:bader_c_factor} gives that the approximation factor difference decreases if $c$ increases. The factor different is defined as the distance between the factor and 1. Clearly if the factor is 1, then our estimation gives the exact correct BC value. This is presented in Theorem \ref{thm:main}. That is, if the number of samples $S=\varepsilon t^{2/3}$, then the factor is less than $b = 1/(\varepsilon t^{1/3}) = t^{1/3}/S$. $b$ should decrease as $S$ increases. It also indicates that the revised Bader's algorithm gives better approximations for sparse graphs. The factor differences are all less than 7\% for both cases.

Figure \ref{fig:bader_sample_factor} gives that the reciprocal of factor difference is almost linearly related to the number of samples. From Theorem \ref{thm:main}, the upper bound $b$ of factor satisfies $b = t^{1/3}/S$, or $1/b \propto S$. Figure \ref{fig:bader_sample_factor} indicates that the actual reciprocal of factor difference seems to be proportional to the number of samplings as well. 

Figures \ref{fig:bader_c_factor} and \ref{fig:bader_sample_factor} together indicates that the revised Bader's algorithm still works well for $1<c<2$.

\subsection{Samplings on Pairs of Vertices}

\begin{figure*}[!t]
\centering
\begin{minipage}[t]{0.3\textwidth}
\centering
\includegraphics[width=\textwidth]{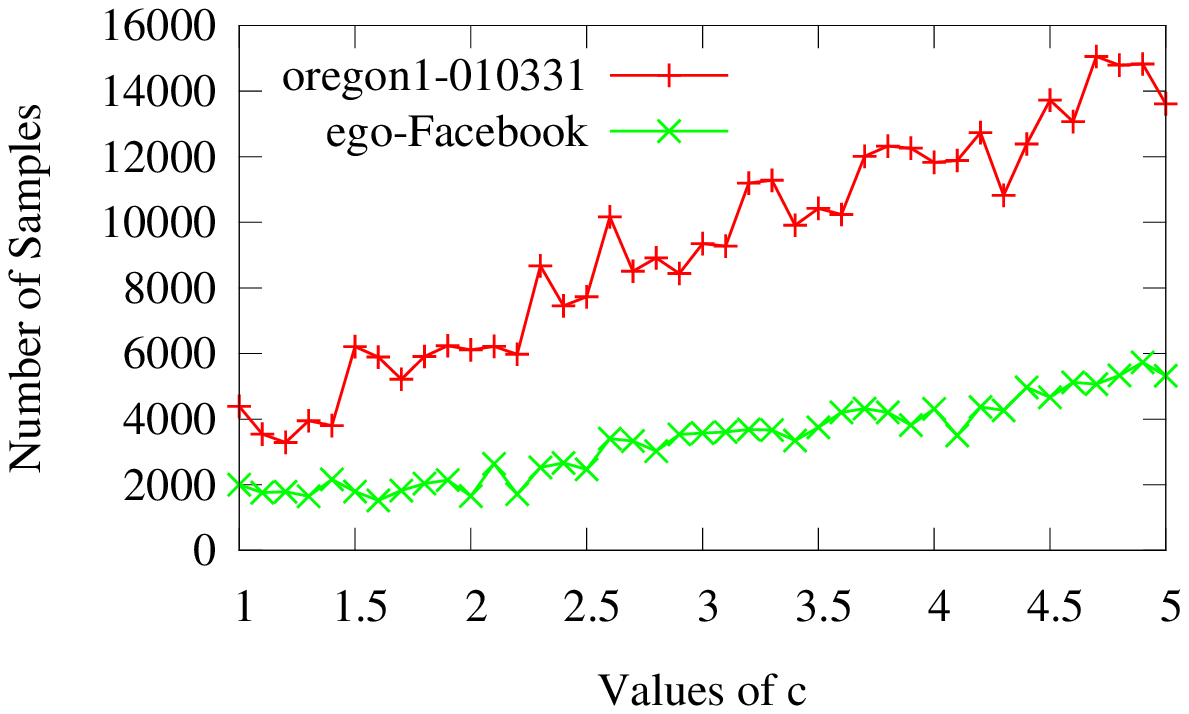}
\caption{$c \sim$ \# samples.}
\label{fig:rion_c_sample}
\end{minipage}
\hspace{10pt}
\begin{minipage}[t]{0.3\textwidth}
\centering
\includegraphics[width=\textwidth]{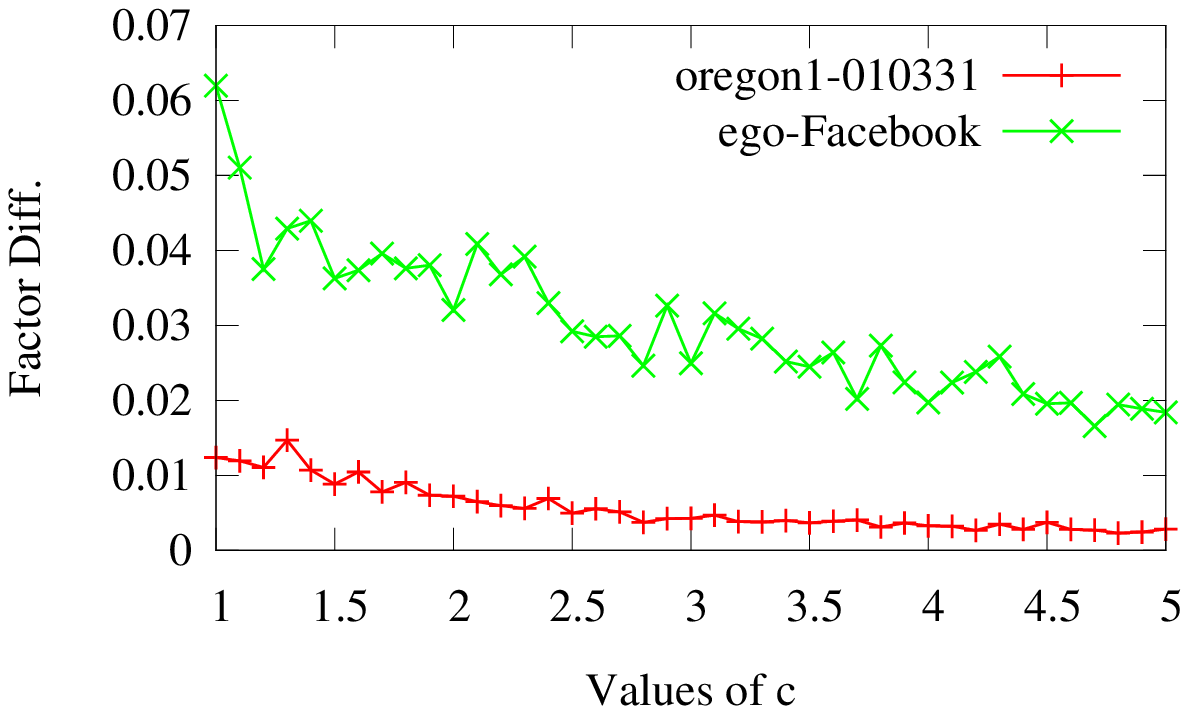}
\caption{$c \sim$ factor diff.}
\label{fig:rion_c_factor}
\end{minipage}
\hspace{10pt}
\begin{minipage}[t]{0.3\textwidth}
\centering
\includegraphics[width=\textwidth]{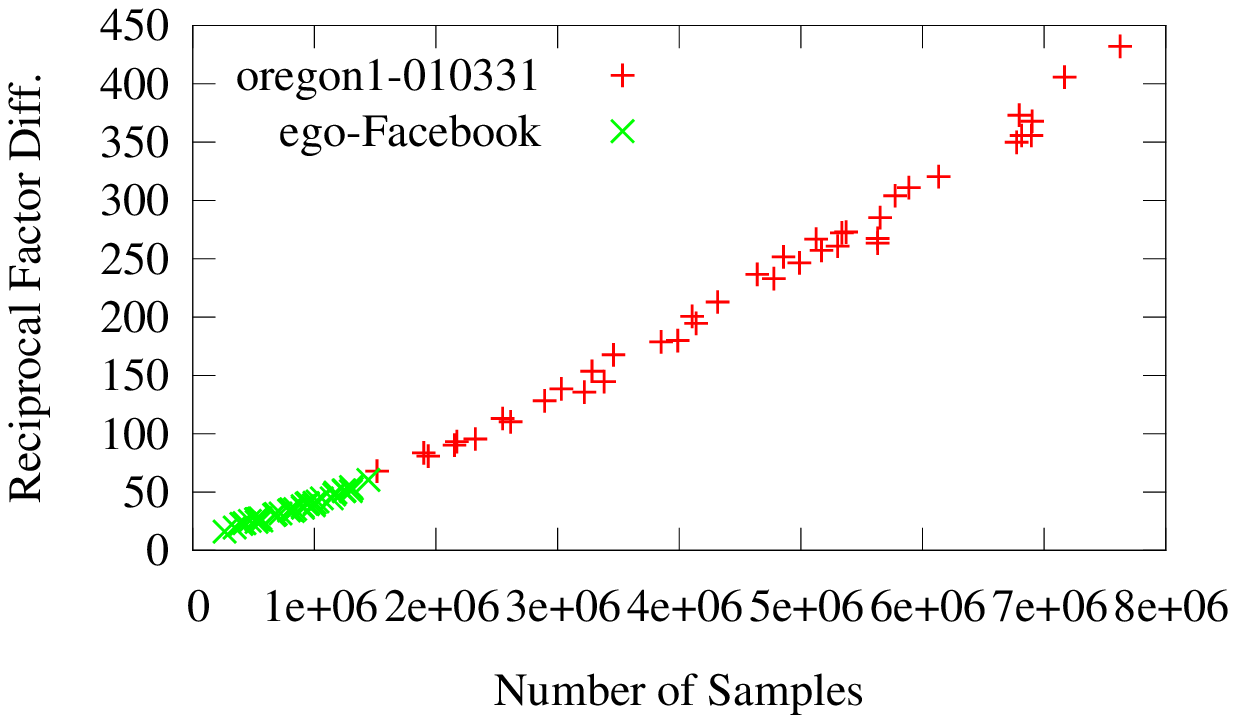}
\caption{\# samples $\sim$ 1/factor diff.}
\label{fig:rion_sample_factor}
\end{minipage}
\end{figure*}

Figure \ref{fig:rion_c_sample} gives that the number of samplings increases if the parameter $c$ in the algorithm increases. The growth is almost linear, since the curves present steady slopes. Similarly to the first algorithm, this is because in the algorithm, higher $c$ gives larger threshold when the loop terminates, allowing more samplings. The number of samplings for graph with more vertices grows faster, since based on Theorem \ref{thm:main2}, the number of samples is proportional to $(n-1)^{1/3}$, where $n$ denotes the number of vertices. The sampled paths are no more than 16000 for both graphs.

Figure \ref{fig:rion_c_factor} gives that the approximation factor difference decreases if $c$ increases. This is presented in Theorem \ref{thm:main2}. That is, if the number of samples $S'=\varepsilon t^{2/3}(n-1)^{1/3}$, then the factor is less than $b' = 1/(\varepsilon (t(n-1))^{1/3}) = t^{1/3}/S'$. $b'$ should decrease as $S'$ increases. It also indicates that our algorithm gives better approximations for sparse graphs. The factor differences are all less than 5\% for both cases, which is slightly better than the previous revised Bader's.

Figure \ref{fig:rion_sample_factor} gives that the reciprocal of factor difference is almost linearly related to the number of samples. From Theorem \ref{thm:main2}, the upper bound $b'$ of factor satisfies $b' = t^{1/3}/S'$, or $1/b' \propto S'$. It also indicates that the actual reciprocal of factor difference seems to be proportional to the number of samplings as well.

\section{Future Works}
One open problem which is revealed in the experiments is that it seems the reciprocal of factor difference is proportional to the number of samplings. The question is whether this conjecture can be shown by theory. Theorems \ref{thm:main} and \ref{thm:main2} both only discuss the upper bound of factors. How to more tightly bound the factors still needs further research.

\section{Conclusions}
In this project, we have revisited Bader's BC estimating algorithm, given more tight bounds on the estimations, and improved the performance of the approximation algorithm, i.e., fewer samples and lower factors. We have also proposed an \emph{adaptive} algorithm with samplings on pairs of vertices. With rigorous reasoning, we have shown that our algorithm can also give bounded approximation factors and number of samplings. To evaluate our results, we have also conducted extensive experiments using real-world datasets. It turns out that both the algorithms can achieve ideal performance. We have verified that our work can improve the current BC approximation algorithm with better parameters, i.e., higher successful probabilities, lower factors and fewer samples.

\section{Acknowledgement}
We are indebted for Matteo Riondato, Vikas Kumar, Rishi Ranjan Singh and Jian Wu for their valuable suggestions and comments on this paper.

\bibliographystyle{acm}
\bibliography{./lit}

\section{Appendix}

\subsection{Proof of Lemma \ref{lem:df}}
\begin{proof}
Note that $\sum_{v_i\in V, v_i\not= v}X_i = A$. Suppose the interval $(0, A)$ is randomly divided into $n$ sub-intervals, each of which represents the length of $X_i$. If the $i$-th sub-interval is larger than $x$, then all the $n-1$ cutting points are distributed outside a sub-interval, whose length is $x$. Formally, we have
$$\p [X_i>x] = \left( 1-\frac{x}{A} \right)^{n-1}$$
and $\p [X_i\leq x] = 1-\p [X_i>x]$ gives the desired cdf. By computing the derivative over $0<x<A$, we have the pdf as desired.
\end{proof}

\subsection{Proof of Corollary \ref{cor:ev}}
\begin{proof}
$$
\begin{aligned}
E[X_i] &= \int_0^A x\frac{n-1}{A} \left( 1-\frac{x}{A} \right)^{n-2} dx \\
&= -x\left( 1-\frac{x}{A}\right)^{n-1}\bigg|_0^A + \int_0^A \left( 1-\frac{x}{A}\right)^{n-1} dx \\
&= -\frac{A}{n} \left( 1-\frac{x}{A} \right)^n\bigg|_0^A = \frac{A}{n}.
\end{aligned}
$$
$$
\begin{aligned}
E[X_i^2] &= \int_0^A x^2\frac{n-1}{A} \left( 1-\frac{x}{A} \right)^{n-2} dx \\
&= -x^2\left( 1-\frac{x}{A}\right)^{n-1}\bigg|_0^A + 2\int_0^A x\left( 1-\frac{x}{A}\right)^{n-1} dx \\
&= \frac{2A}{n}\int_0^A \left( 1-\frac{x}{A}\right)^{n} dx = \frac{2A^2}{n(n+1)}.
\end{aligned}
$$
$$Var[X_i] = E[X_i^2] - E[X_i]^2 = \frac{2A^2}{n(n+1)} - \frac{A^2}{n^2} = \frac{n-1}{n^2(n+1)}A^2.$$
\end{proof}
Since $0\leq A \leq n^2$, it is clear our model agrees with \cite{bader2007approximating}'s bounds on $E[X_i]$ and $Var[X_i]$.
However for $E[X_i^2]$, our model gives the upper bound: $E[X_i^2] \leq 2A$. This can be explained: it is possible for $Var[X_i]$ to reach $A$ very closely (e.g., when $A$ is very close to $n^2$), and since $E[X_i^2] = Var[X_i] + E[X_i]^2 = Var[X_i] + A^2/n^2$, it is possible that $Var[X_i]$ exceeds $A$. Luckily neither this paper nor Bader's work uses this bound on $E[X_i^2]$ in the rest reasoning.

\subsection{Proof of Lemma \ref{lem:1}}
\begin{proof}
Note that $BC(v) = A \leq n^2$. Then we have
$$
\begin{aligned}
&\p[\sum_{i=1}^k X_i \geq cn] = \p[\sum_{i=1}^k \left( X_i - \frac{A}{n} \right) \geq cn-\frac{kA}{n}] \\
= & \p[\sum_{i=1}^k \left( X_i - \frac{A}{n} \right) \geq cn-\varepsilon n(\frac{A}{n^2})^{\frac{1}{3}}] \\
\leq & \p[\sum_{i=1}^k \left( X_i - \frac{A}{n} \right) \geq cn-\varepsilon n] &(A\leq n^2)\\
\leq & \p[\bigvee_{i=1}^k\left( X_i - \frac{A}{n} \right) \geq \frac{c-\varepsilon}{k}n] &\textrm{(at least one $X_i$ exceeds the average)} \\
\leq & \sum_{i=1}^k \p[\left( X_i - \frac{A}{n} \right) \geq \frac{c-\varepsilon}{k}n] &\textrm{(union bound)} \\
\leq & \sum_{i=1}^k \frac{k^2}{(c-\varepsilon)^2n^2} Var[X_i] &\textrm{(Chebychev's inequality)}\\
= & \frac{k^2}{(c-\varepsilon)^2n^2} \sum_{i=1}^k \frac{n-1}{n^2(n+1)}A^2 &\textrm{(Corollary \ref{cor:ev})} \\
= & \frac{\varepsilon^3 (n-1)}{(c-\varepsilon)^2 (n+1)} \leq \frac{\varepsilon^3}{(c-\varepsilon)^2}.
\end{aligned}
$$
\end{proof}

\subsection{Proof of Lemma \ref{lem:2}}
\begin{proof}
$$
\begin{aligned}
&\p[\bigg|\frac{n}{k}\sum_{i=1}^k X_i - A\bigg|\geq dA]
= \p[\bigg|\sum_{i=1}^k X_i - \frac{kA}{n}\bigg|\geq \frac{kdA}{n}] \\
= &\p[\bigg|\sum_{i=1}^k \left( X_i - \frac{A}{n}\right) \bigg|\geq \frac{kdA}{n}]
\leq \frac{n^2}{k^2d^2A^2}k Var[X_i] & \textrm{(Chebychev's inequality)}\\
= &\frac{n-1}{(n+1)kd^2} \leq \frac{1}{kd^2} = \frac{1}{\varepsilon d^2} \left(\frac{A}{n^2}\right)^{\frac{2}{3}}.
\end{aligned}
$$
\end{proof}

\subsection{Proof of Lemma \ref{lem:3}}
\begin{proof}
$$
\begin{aligned}
&\p[\sum_{i=1}^k Y_i \geq cn] \\
= &\p[\sum_{i=1}^k \left( Y_i - \frac{A}{n(n-1)} \right) \geq cn-\frac{kA}{n(n-1)}] \\
= & \p[\sum_{i=1}^k \left( Y_i - \frac{A}{n(n-1)} \right) \geq cn-\varepsilon n(\frac{A}{n^2(n-1)})^{\frac{1}{3}}] \\
\leq & \p[\sum_{i=1}^k \left( Y_i - \frac{A}{n(n-1)} \right) \geq cn-\varepsilon n] &(A\leq n^2)\\
\leq & \p[\bigvee_{i=1}^k\left( Y_i - \frac{A}{n(n-1)} \right) \geq \frac{c-\varepsilon}{k}n] &\textrm{(at least one $Y_i$ exceeds the average)} \\
\leq & \sum_{i=1}^k \p[\left( Y_i - \frac{A}{n(n-1)} \right) \geq \frac{c-\varepsilon}{k}n] &\textrm{(union bound)} \\
\leq & \sum_{i=1}^k \frac{k^2}{(c-\varepsilon)^2n^2} Var[Y_i] &\textrm{(Chebychev's inequality)}\\
= & \frac{k^2}{(c-\varepsilon)^2n^2} \sum_{i=1}^k \frac{n(n-1)-1}{n^2(n-1)^2(n(n-1)+1)}A^2 &\textrm{(Corollary \ref{cor:ev2})} \\
= & \frac{\varepsilon^3 (n(n-1)-1)}{(c-\varepsilon)^2 (n(n-1)+1)} \leq \frac{\varepsilon^3}{(c-\varepsilon)^2}. & \textrm{(Jensen's inequality)}
\end{aligned}
$$
\end{proof}
Alternatively one can use the bound on variance from \cite{bader2007approximating}: $Var[Y_i] < A$. Analogous computation can show the upper bound is still the same.

\subsection{Proof of Lemma \ref{lem:4}}
\begin{proof}
$$
\begin{aligned}
&\p[\bigg|\frac{n(n-1)}{k}\sum_{i=1}^k Y_i - A\bigg|\geq dA]\\
= &\p[\bigg|\sum_{i=1}^k Y_i - \frac{kA}{n(n-1)}\bigg|\geq \frac{kdA}{n(n-1)}] \\
= &\p[\bigg|\sum_{i=1}^k \left( Y_i - \frac{A}{n(n-1)}\right) \bigg|\geq \frac{kdA}{n(n-1)}]
\leq \frac{n^2(n-1)^2}{k^2d^2A^2}k Var[Y_i] & \textrm{(Chebychev's inequality)}\\
= &\frac{n(n-1)-1}{(n(n-1)+1)kd^2} \leq \frac{1}{kd^2} = \frac{1}{\varepsilon d^2} \left(\frac{A}{n^2(n-1)}\right)^{\frac{2}{3}}. & \textrm{(Jensen's inequality)}
\end{aligned}
$$
\end{proof}

\end{document}